\documentclass[a4paper]{article}
\usepackage{amsmath,amssymb,amsthm}
\usepackage[T1]{fontenc}
\usepackage{enumerate}
\usepackage{graphicx}
\usepackage{dsfont}
\usepackage{float}
\usepackage[utf8]{inputenc}
\usepackage{color}
\definecolor{light-gray}{gray}{0.95}
\usepackage{listingsutf8}
\usepackage{url}
\usepackage{authblk}

\newcommand{\ew}[1]{\mathbb{E}\left[#1\right]}
\newcommand{\p}[1]{\mathbb{P}\left[#1\right]}

\newcommand{\rr}{\mathbb{R}}

\theoremstyle{plain}

\newtheorem{Proposition}{Proposition}[section]

\theoremstyle{definition}
\newtheorem{Definition}{Definition}[section]

\theoremstyle{remark}
\newtheorem{Remark}{Remark}[section]

\author[1]{Rainer Haidinger~\thanks{rainer.haidinger@gmail.com}} 
\author[1]{Richard Warnung\thanks{richard.f.warnung@gmail.com}}
 \affil[1]{{\small Risikomanagement, Raiffeisen Kapitalanlage-Gesellschaft m.b.H.
Schwarzenbergplatz 3, 1010 Vienna, Austria  \footnote{The contents of this paper are the authors' sole responsibility. They do not necessarily represent the
views of Raiffeisen Kapitalanlage-Gesellschaft m.b.H. Both authors thank the members of the risk
management department of Raiffeisen Kapitalanlage-Gesellschaft m.b.H. for fruitful discussions.}
\newline
Financial and Actuarial Mathematics, Vienna University of Technology,
Wiedner Hauptstrasse 8-10/105-1, 1040 Vienna, Austria}}

\title{Risk Measures in a Regime Switching Model Capturing Stylized Facts}
\date{\today}
\begin{document}
\maketitle

\begin{abstract}
We pick up the regime switching model for asset returns introduced by Rogers and Zhang. The calibration involves various markets including implied volatility in order to gain additional predictive power. We focus on the calculation of risk measures by Fourier methods that have successfully been applied to option pricing and analyze the accuracy of the results. 
\end{abstract}

\section{Intro}

This article is based on the innovative approach for modelling asset returns introduced by Chris Rogers and Liang Zhang in their paper "An asset return model capturing stylized facts"~\cite{rogers}. The model captures several stylized facts of asset returns which is achieved by modelling returns with a hidden two-state Markov chain representing states of the economy. Conditional on these states returns are modelled by generalized hyperbolic distributions. We shortly describe the model and its calibration in Section~\ref{sec:model}.

We apply the model to equity, bond and commodity markets. Furthermore we include the VDAX index, an index that represents implied volatility for the future $45$ days, in the calibration procedure and hope to increase the predictive power of the model by this forward looking ingredient. The final aim is to calculate value-at-risk and expected shortfall in this framework by Fourier techniques from option pricing theory~\cite{fft,cont2004financial,Lewis:simple} in Section~\ref{sec:riskmeasures}.

In the following we recall some stylized facts exhibited by financial time series across a wide range of instruments, markets and time periods. 

Usually the distribution of daily log-returns is leptokurtic which means that it shows fat tails and a high peak at its median. Since the normal distribution is mesokurtic as its excess-kurtosis is zero, it is not fully adequate for modelling such returns. In Figure~\ref{ghyp_normal} the fat tails in the return histogram and the bad fit of the normal distribution to the data is clearly visible. The returns were calculated for daily closing prices of the DAX Index in the period from 18 November 2005 to 30 December 2009. The second plot shows a QQ-Plot of the data versus the theoretical quantiles of the normal distribution and again the deviation is visible. The Jarque-Bera statistic confirms the departure  from normality at a one percent level of significance\footnote{p-value < 2e-16.}  .

\begin{figure}[H]
\begin{center}
\includegraphics[width=\textwidth]{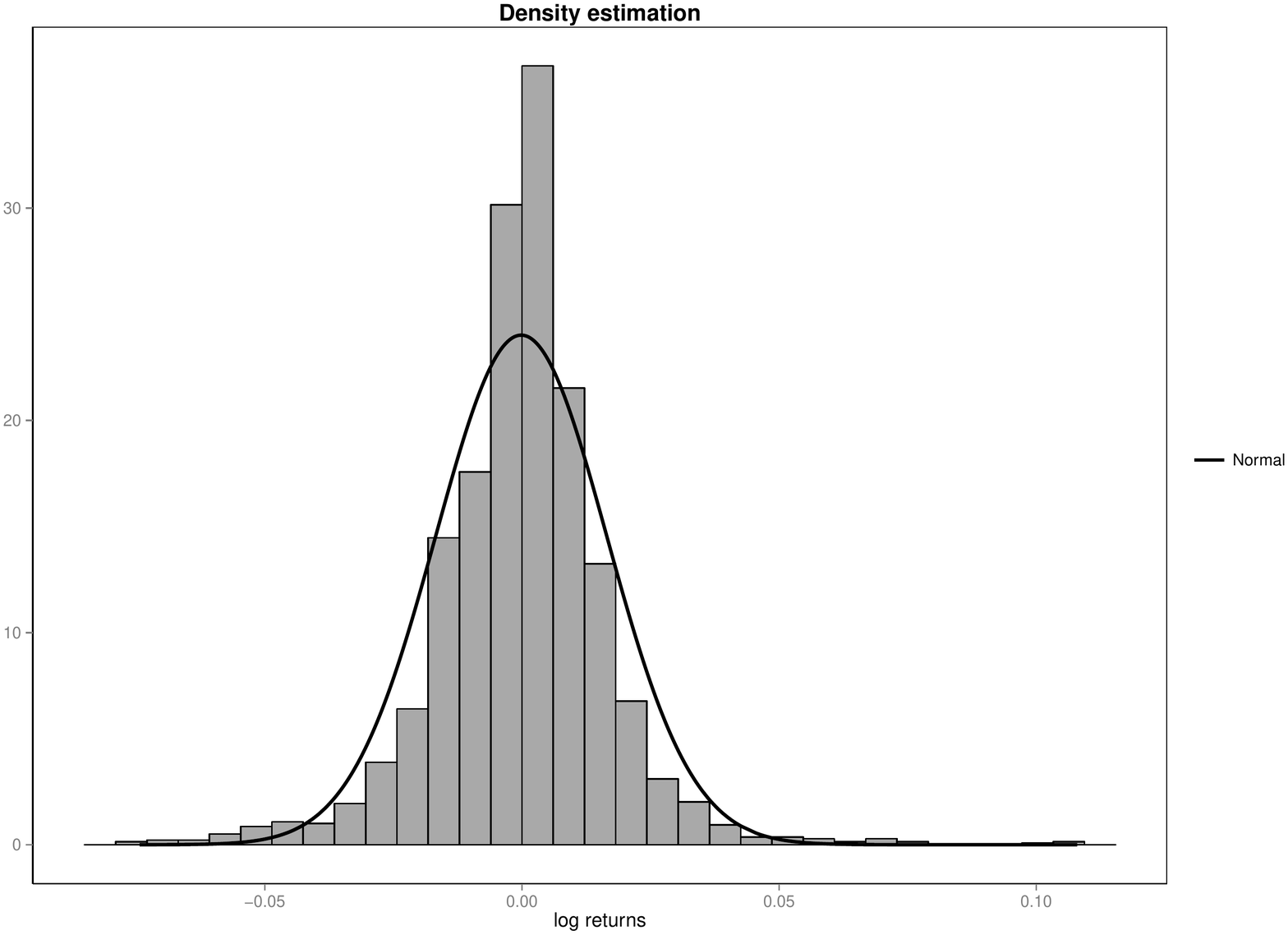}
\caption{Histogram of daily DAX log returns}
  \label{ghyp_normal}
\end{center}
\end{figure}

\begin{figure}[H]
\begin{center}
\includegraphics[width=\textwidth]{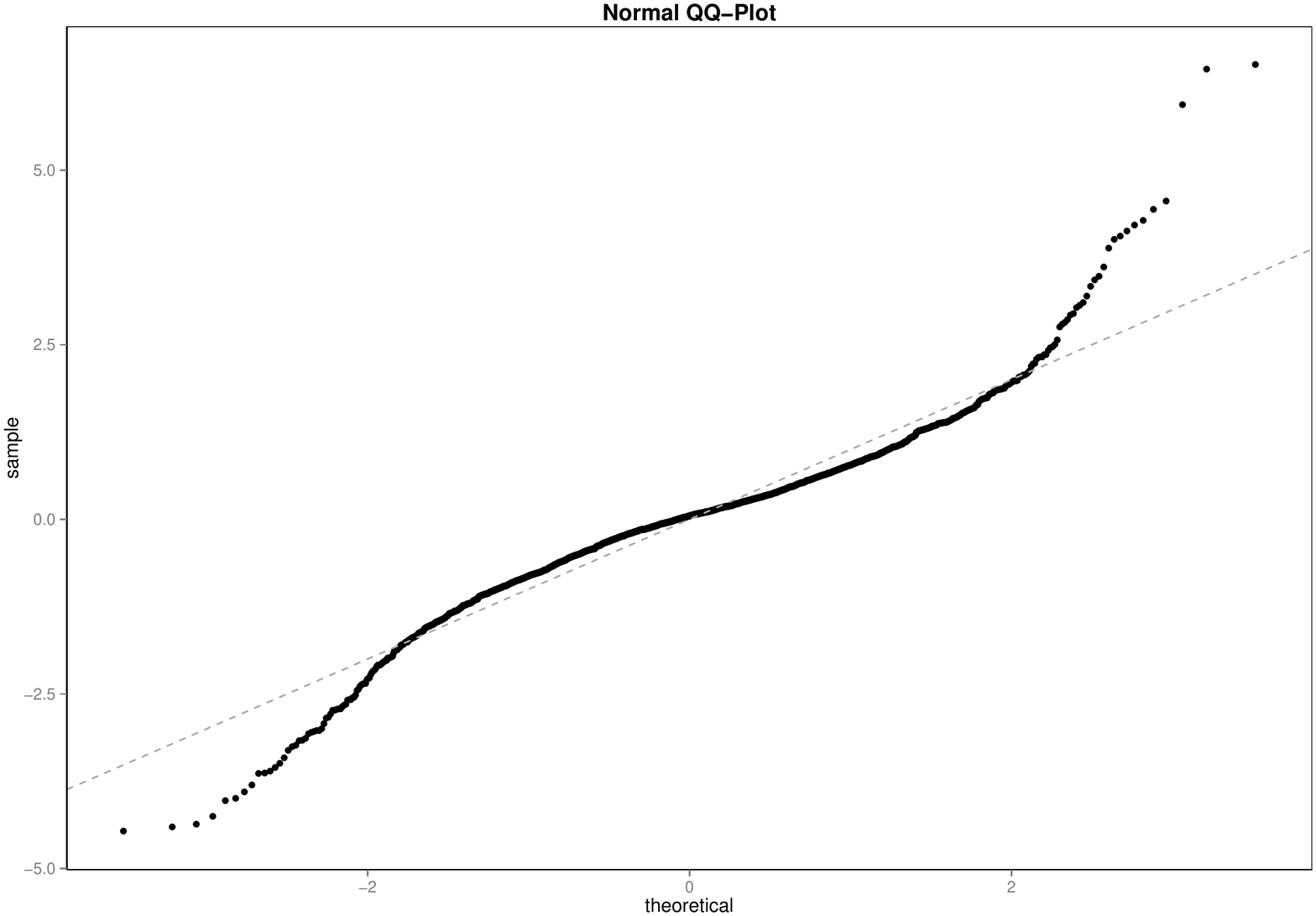}
\caption{QQ-Plot of daily DAX log returns}
  \label{qq_plot}
\end{center}
\end{figure}

Furthermore the autocorrelation of asset returns is usually observed to be insignificant while the autocorrelation of absolute returns is significantly different from zero up to surprisingly high lags. This is known as volatility clustering. Large movements tend to be followed by large movements, irrelevant of the sign. In Figure~\ref{acf} we see a plot of the autocorrelation function (ACF) of DAX returns with no significant autocorrelation whereas the plot of the ACF of absolute returns shows high autocorrelations and a slow decay even at high lags.

\begin{figure}[H]
\begin{center}
\includegraphics[width=\textwidth]{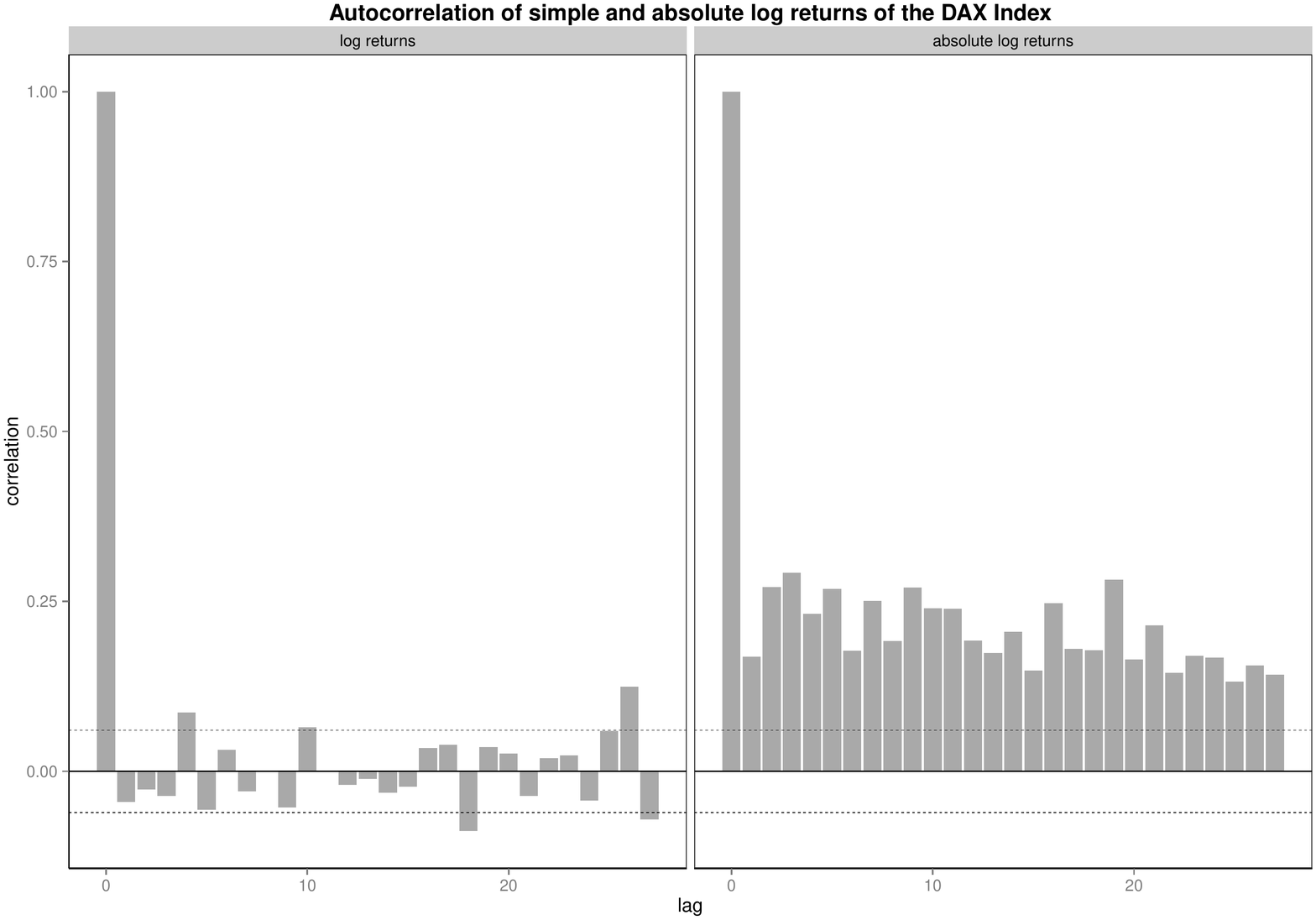}
\caption{Autocorrelation structure daily DAX log returns}
  \label{acf}
\end{center}
\end{figure}

Finally, there is an asymmetry between large losses and large gains, large, sudden draw-downs can be observed. Gains are usually made slowly and in a less sharp pattern. As a consequence the distribution is skewed. This fact applies to stocks but not necessarily to indices. In Figure~\ref{ghyp_normal} we see that the distribution of the returns of the DAX index is fairly symmetric.

Rogers and Zhang~\cite{rogers} develop a model that captures all these facts. The main idea is to use a hidden Markov model with a non-Gaussian distribution of returns given the states. In the following we shortly describe their model and its calibration.

\section{The Hidden Markov Model of Rogers and Zhang}
\label{sec:model}
\subsection{General Properties}

Hidden Markov models, in finance sometimes applied as regime switching models, are  mixture models for financial time series where the distribution of the series is determined by an unobservable Markov chain. Whereas the chain is unobservable, the returns are observed and conditional statements can be done.

In this model we consider an ergodic two-state first order Markov chain $(\xi_n)_{n \in \mathbb{N}}$ . Let  $P = \{p_{ij}\} \in \rr^{ 2}$ denote the transition matrix. The $k$-step transition probabilities are denoted by $p^k_{ij}=\p{\xi_n=j|\xi_{n-k}=i}$ and $P^{(k)}=\{p^k_{ij}\}:=P^k$ denotes the corresponding $k$-step transition matrix.

For modelling the returns, two sequences $\left(X_n^i \right)_{n \in \mathbb{N}}$ for $\ i \in \{1,2\}$ independent of $(\xi_n)_{n \in \mathbb{N}}$ are defined. In each state  $\ i \in \{1,2\}$ we have $X_n^i \sim F_i $ for $n \in \mathbb{N}$. The concrete choice of the distributions $F_i,\ i \in \{1,2\}$, will be revealed further below.

Based on these definitions the returns are modelled by
\begin{align}
r_n=\sum_{i=1}^{2} \mathbf{1}_{\left\{\xi_n=i\right\}}X_n^i, \quad\text{for } n \in \mathbb{N}.
\end{align}
As  $(\xi_n)_{n \in \mathbb{N}}$ is ergodic we can determine the stationary distribution $\pi=\begin{pmatrix}\pi_1,\pi_2\end{pmatrix}'$.
The stationary distribution is given by the normalized left eigenvector of $P$, so the calculation of $\pi$ for the two-state model is easily done and we get
\begin{align*}
\pi=\begin{pmatrix}\frac{1-p_{22}}{2-p_{11}-p_{22}} \\\frac{1-p_{11}}{2-p_{11}-p_{22}} \end{pmatrix}.
\end{align*}

The model incorporates the stylized facts discussed above. To show this we cite the formula for the covariance structure of $\left(|r_n|\right)_{n\in \mathbb{N}}$ (see~\cite{ryden}).
\begin{Proposition} Let $g$ be a function of the return, such that $\ew{g(r_n)|\xi_n=S_i}$, the expected value given the state, exists for $\ i \in \{1,2\}$. Then
the autocovariance of $g\left(r_n\right)_{n\in \mathbb{N}}$  is given by
\begin{align}\label{eq.acf}
\mathbb{COV}\left[g(r_n),g(r_{n+k})\right]=\boldsymbol{\pi}\mathbf{G}\mathbf{P}^{(k)}\mathbf{G}\mathbf{1}-\left(\boldsymbol{\pi}\mathbf{G}\mathbf{1} \right)^2, \ k=1,2,\ldots 
\end{align}
where $\mathbf{G}:=\textrm{diag}(G_1,G_2)$ and $G_i:=\ew{g(r_n)|\xi_n=S_i},\ i \in \{1,2\} $.
\end{Proposition} 

One stylized fact mentioned is that  log returns do not exhibit significant autocorrelation.
The following  proposition of Rogers and Zhang~\cite{rogers} assures that the model captures this property.
\begin{Proposition}
Let $\mu_i:=\ew{r_n|\xi_n=i}$  for $i \in \{1,2\}$. If $\mu_i = \mu$ for $i \in \{1,2\}$ then

\begin{align}
\ew{r_nr_{n+k}}=\mu^2
\end{align}
for any $k>0$ and $n \in \mathbb{Z}$.
\end{Proposition}
It is worth noting that the proof of this statement does not require that the Markov chain has only two states but it is true for any finite number of states as long as the expected returns are equal. Rogers and Zhang~\cite{rogers} find that a two state Markov chain suffices to reproduce the stylized facts mentioned. Furthermore the economic interpretation as riskier and less risky times is clear in this set-up. Our findings in the calibration example in Subsection~\ref{sub:calibration} support this view.

So far we have assumed that $X_n^i \sim F_i \ \ \forall \ n$ and $i \in\{1,2\}$, where $F_i$ is a general distribution function. In~\cite{rogers} the class of Generalized Hyperbolic (GHYP) distributions is chosen. 
There are several alternative parametrizations for the GHYP family, each with different parameter ranges. The following parametrization is the most widely used in literature.

\begin{Definition}
The univariate density of a GHYP $\left(\alpha,\beta,\lambda,\delta,\mu\right)$ is given by
\begin{align}
f_{X}(x)&=\frac{(\alpha ^2-\beta^2)^{\lambda/2}}{(2\pi)\alpha^{\lambda -\frac{1}{2}}\delta^\lambda K_{\lambda} (\delta \sqrt{\alpha ^2-\beta^2}) } \times \frac{K_{\lambda-\frac{1}{2}} (\alpha \sqrt{\delta ^2+(\mathbf{x}-\mu)^2})e^{\beta(x-\mu)}}{\sqrt{\delta ^2+(\mathbf{x}-\mu)^2}^{\frac{1}{2}-\lambda}},
\end{align}
\noindent where $K_\lambda$ is a modified Bessel function of the second kind.  
\end{Definition}

The generalized hyperbolic distribution is a flexible class  as it covers a wide range of well-known distributions as limiting or subclasses.
 One can obtain the normal distribution as a limiting case for $\delta \rightarrow \infty$ and $\delta/\alpha \rightarrow \sigma^2$.
For $\lambda=1$ these are the hyperbolic distributions, whereas for $\lambda=-1/2$ we get the normal inverse Gaussian distribution.
Another important limiting case is the Student-t distribution, which is a limit for $\lambda<0$ and $\alpha=\beta=\mu=0$. Furthermore, using the GHYP distribution we can model the observed heavy tails. In particular
\begin{align*}
\textrm{GHYP} \sim \left|x\right|^{\lambda-1}e^{(\pm\alpha+\beta)x}, \  \textrm{for} \ x\rightarrow \pm \infty.
\end{align*}

Recall the moment generating function in the chosen parametrization.
\begin{Definition}
The moment generating function (MGF) of a generalized hyperbolic distribution is
\begin{align*}
MGF(t)=e^{t\mu }\left( \frac{\alpha ^2-\beta ^2}{\alpha^2-(\beta+t)^2} \right)^{\lambda /2} \frac{K_{\lambda} (\delta \sqrt{\alpha ^2-(\beta+t)^2})}{K_{\lambda} (\delta \sqrt{\alpha ^2-\beta^2})}
\end{align*}
\noindent for $|\beta+t|<\alpha$.
\end{Definition}
In Subsection~\ref{Subsec:VaR} we will apply Fourier techniques to calculate the VaR and the ES.
It can be shown that the MGF is real analytic, therefore MGF is a holomorphic function for complex $z$ with $|z|<\alpha-\beta$. So the characteristic function can be obtained through $\phi(t)=\textrm{MGF}(it)$.

In the plot in Figure~\ref{ghyp_ghyp} we illustrate the application of GHYP in the case of log returns of the DAX index. In the histogram a fitted generalized hyperbolic distribution is superimposed together with a fitted normal distribution. The superior fit of GHYP in the center as well as in the tails is plain to see. 

\begin{figure}[H]
\begin{center}
\includegraphics[width=\textwidth]{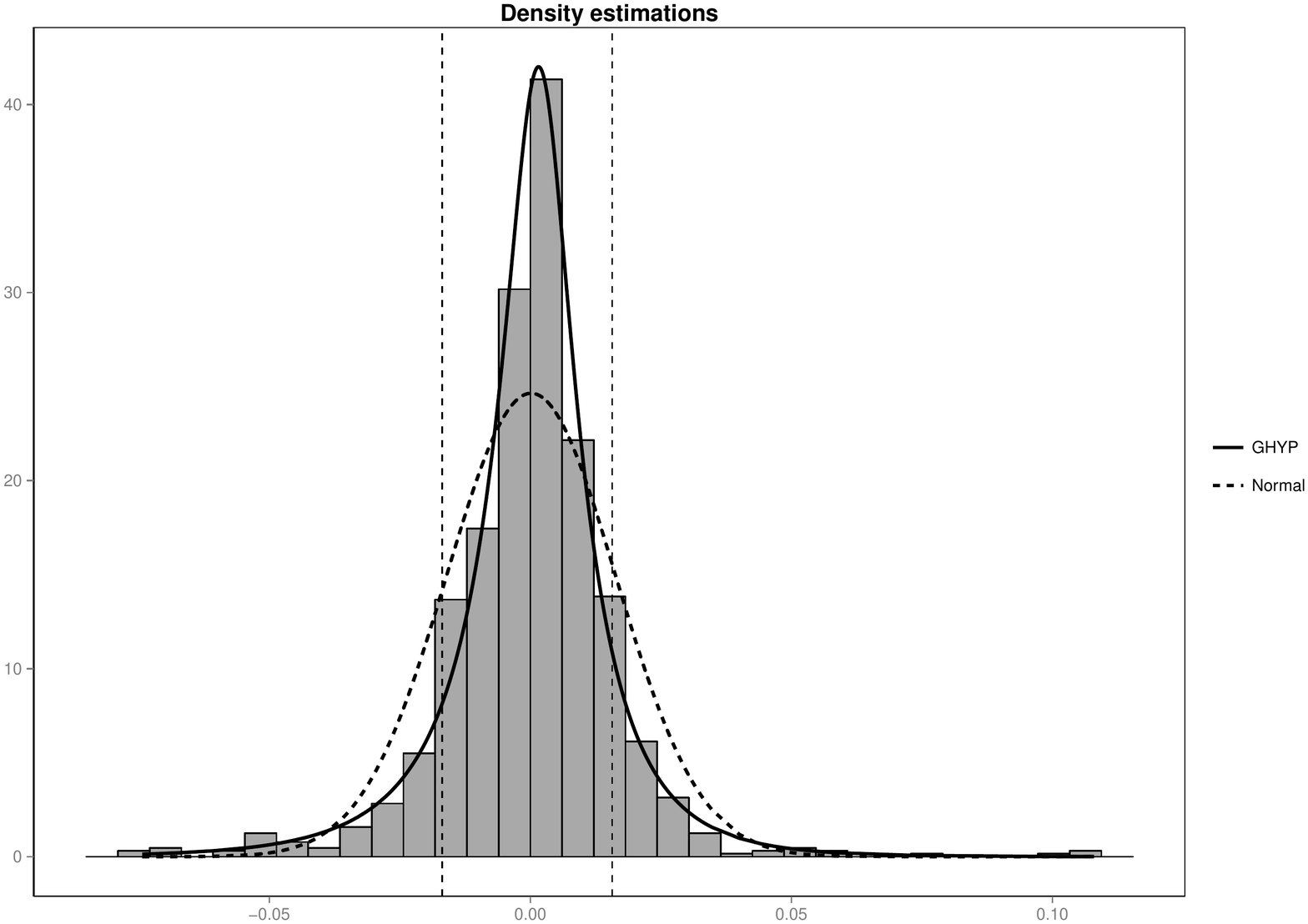}
\caption{Histogram of daily DAX log returns}
  \label{ghyp_ghyp}
\end{center}
\end{figure}

\subsection{Calibration example} \label{sub:calibration}

Given the fact that the model captures all the stylized facts mentioned, it is remarkable that it still can be calibrated using maximum likelihood. Additionally to the likelihood function a term that penalizes a deviation of the theoretical autocovariance function of absolute returns given by~\eqref{eq.acf} from the observed values is incorporated (for details we refer to~\cite{rogers}). We have implemented the model and its calibration in R~\cite{RManual} using the functions for the GHYP distribution provided in~\cite{ghyp}.

We calibrated the model to four indices. Note that as in~\cite{rogers} the Markov chain is one and the same for all indices. This will help us to derive an economic interpretation of the states. Except for the commodities index we focused on European indices in order to avoid closing time differences (see for example~\cite{war}). We considered the following indices:

\begin{description}
\item[Stocks]  For stocks the DAX index, the German stock index was chosen. The DAX index comprises of 30 major German companies and is one of the most important stock indices in Europe.
\item[Bonds] As a representation for bonds as less risky investment opportunity the REX index, the German Rentenindex, was selected. The REX measures the performance of $30$ German government bonds. Note that the REX is not investable but it still serves as a good proxy for the bond market.

\item[Commodities] The Dow Jones UBS Commodity Index (DJUBS) was considered. The Index is a futures based index which includes the following commodities markets: energy, precious metals, agricultural goods, industrial metals, and live stock.
\item[Volatility] The VDAX the German volatility index was also included as it tries to capture implied volatility of the DAX in a forward looking way. The VIX, the analogous to the VDAX  is known as a `fear' index. Although the  returns of a volatility index are of a quite different nature than stock index returns (compare estimated densities in of DAX returns resp. VDAX in Figure~\ref{fitdensity}) we included the VDAX in order to get more predictive power for the risk measures.
\end{description}

The time series consists of 1609 daily returns from 18 November 2005 to 27 June 2012 \footnote{The data was obtained via yahoo finance.}. This data set was divided into two parts one from 18 November 2005 to 30 December 2009 as calibration period for the model and from 04 January 2011 to 27 June 2012 as out-of-sample period used for backtesting.
Figure~\ref{timeseries} shows the movement of the four indices for the whole time period from 18 November 2005 to 27 June 2012. We can clearly see the spike in the VDAX beginning in late 2008 and the corresponding decline in the DAX and the DJUBS.
\noindent In the last year the Fukushima crisis is visible as a slight downward movement and of course the European sovereign-debt crisis with the stock market crash beginning in August 2011. 

\begin{figure}[H]
\begin{center}
\includegraphics[width=\textwidth]{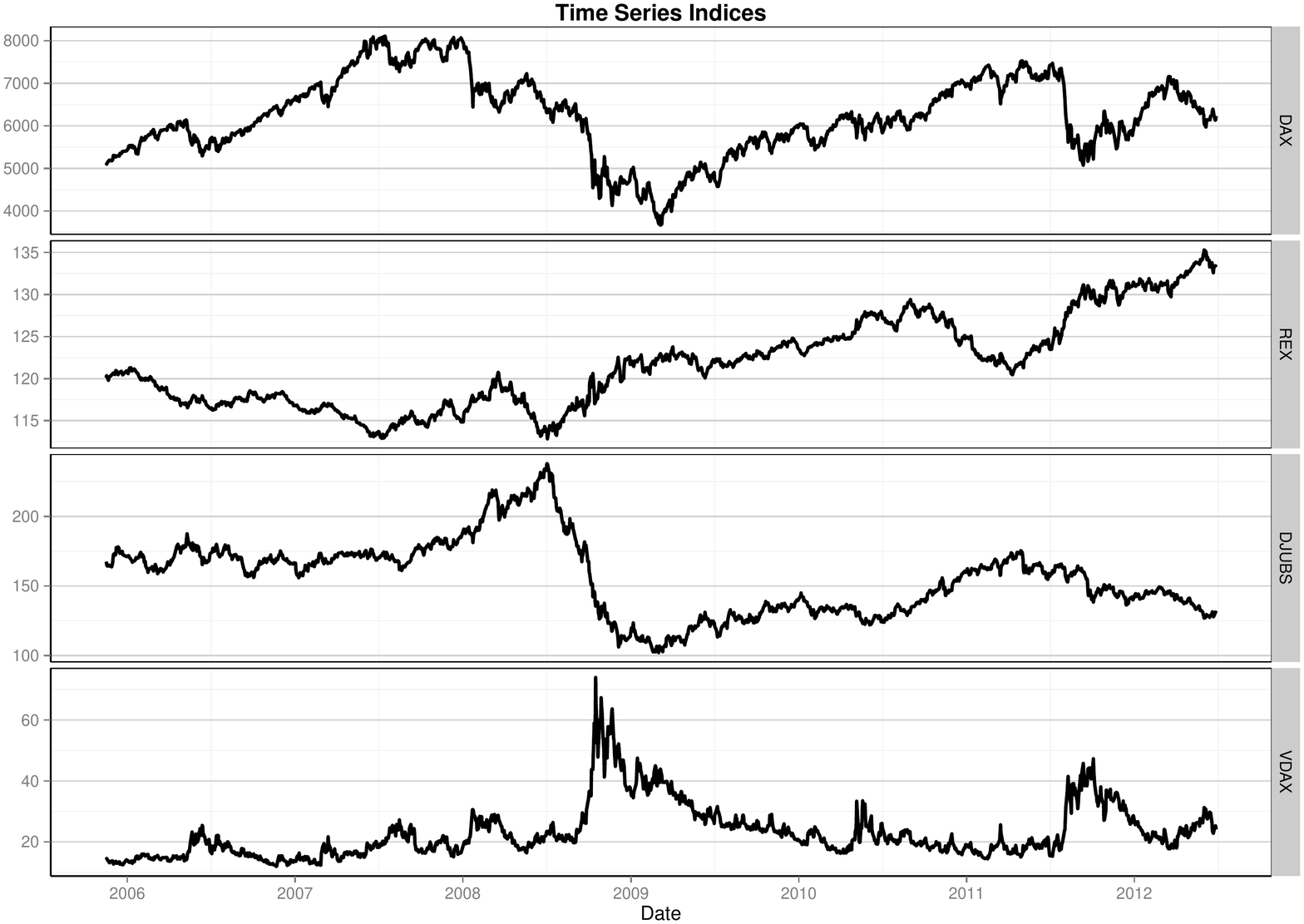}
\caption{Market performance 2005 - 2012}
  \label{timeseries}
\end{center}
\end{figure}

\begin{Remark}
To estimate the parameters in this model with two states, one needs a time series long enough as there are 38 parameters to estimate, namely ten parameters for a single asset plus two transition probabilities. As we require the same means in both states there are four parameters less to estimate in total. 
\end{Remark}

The next figures show the results of the maximum likelihood estimation. In Figure~\ref{acfestdax} the observed autocovariance structure versus the autocovariance  in the estimated model is plotted. We see that we get a fairly good fit. 
Figure~\ref{fitdensity} shows the estimated density functions for the considered indices. The densities are quite symmetric (except for the VDAX), as expected. Furthermore all fitted distributions exhibit fat tails and the distribution in one state is clearly riskier (more leptokurtic) than in the other state. Therefore we can conclude that the model has an economic meaning in the sense that one state represents an optimistic and the other state a pessimistic market. 
Another property is that if we are in the optimistic state for the stock index we are in the pessimistic state for bonds and vice versa. Which is again our natural intuition, that in optimistic upward moving market we tend to invest in stocks.

The optimistic state for the bond index is the same as the optimistic state for the commodity index which is surprising. If we think of simple correlation the commodity index should go together with the stock index. Maybe this fact is due to the broadness of the commodity index.

\begin{figure}[H]
\begin{center}
\includegraphics[width=\textwidth]{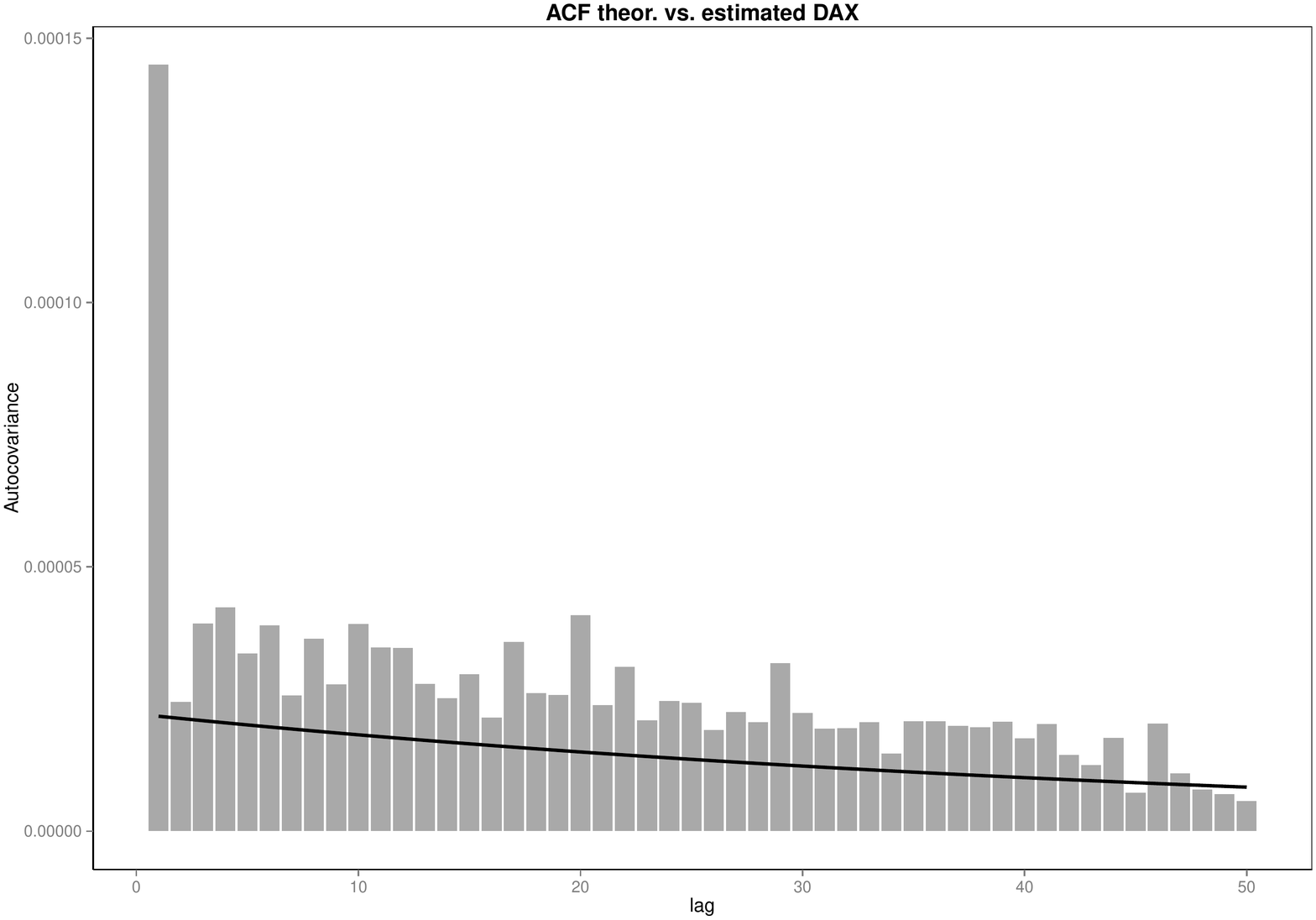}
\caption{ACF theoretical vs. observed DAX index}
  \label{acfestdax}
\end{center}
\end{figure}

\begin{figure}[H]
\begin{center}
\includegraphics[width=\textwidth]{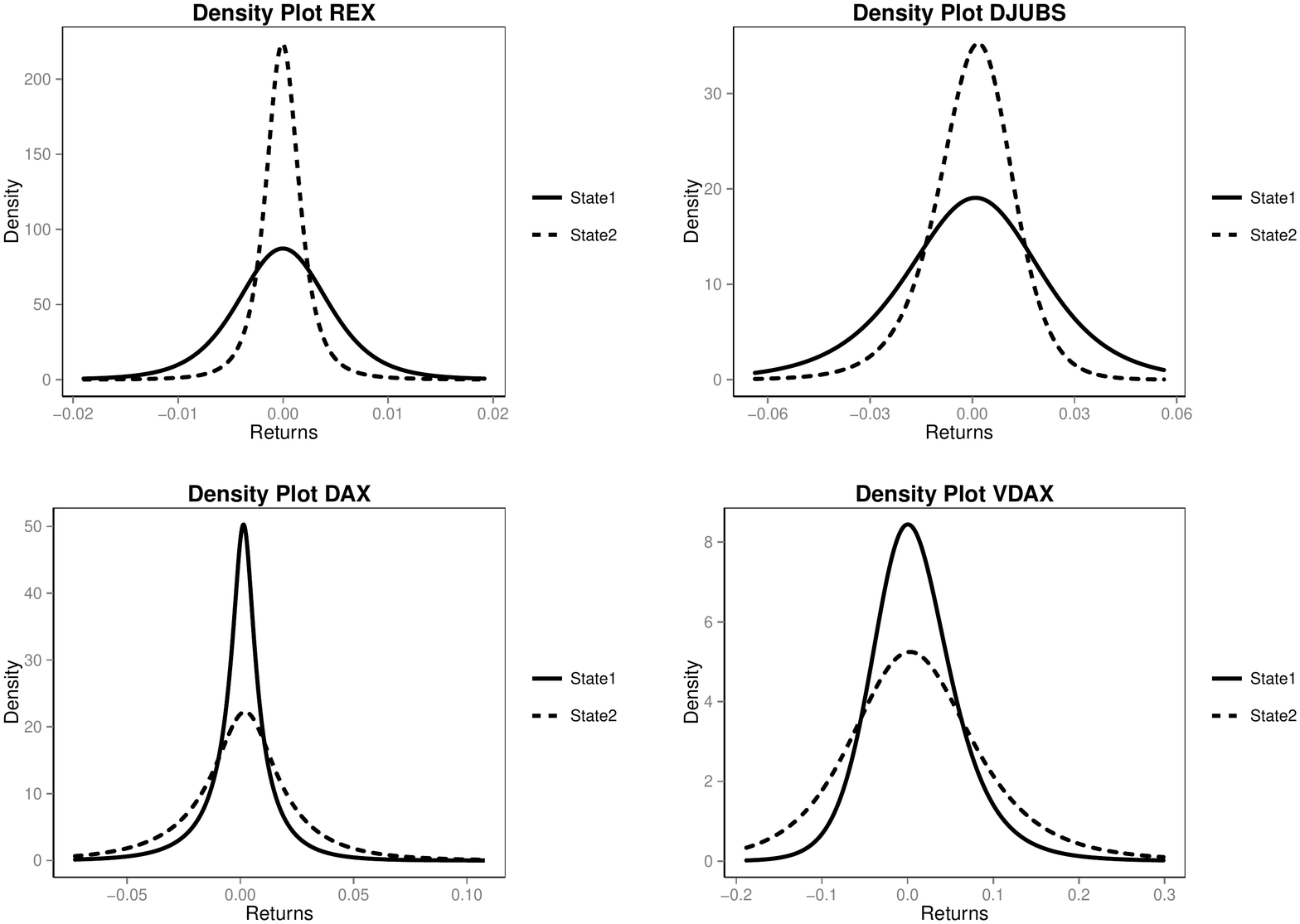}
\caption{Density function for both states for the four indices}
  \label{fitdensity}
\end{center}
\end{figure}

Having calibrated the model to the data we turn to applications. Rogers and Zhang present a number of applications~\cite{rogers} whereas we focus on an application to risk measures.

\section{Risk measures}
\label{sec:riskmeasures}

The first object that we want to address (see also~\cite{rogers}) is the posterior probability of being in a certain state. This means that having observed a certain return it tells us how our view of the current state changes.

\begin{Definition}
The posterior probability $p_n = (p_n^1,p_n^2)^T$  at time $n$ can be calculated by

\begin{align}\label{eq:post}
p_n=\frac{p_{n-1}PF(r_n,\theta_1,\theta_2)}{p_{n-1}PF(r_n,\theta_1,\theta_2)\mathbf{1}},
\end{align}
where\begin{align*}
F(r;\theta_1,\theta_2)=diag\left(\prod_{j=1}^mf(r^j;\theta_1^j),\prod_{j=1}^mf(r^j;\theta_2^j)\right),
\end{align*}
$m$ is the number of markets, $f$ the GHYP density and $\theta_{i}^{j}$ with $i \in \{1,2\}$ and $j \in 1,\ldots,m$ denote the parameters of the given market in the respective state.

\end{Definition}
The next plot shows the posterior probability of being in the optimistic state for the DAX index together with the performance of the index.

\begin{figure}[H]
\begin{center}
\includegraphics[width=\textwidth]{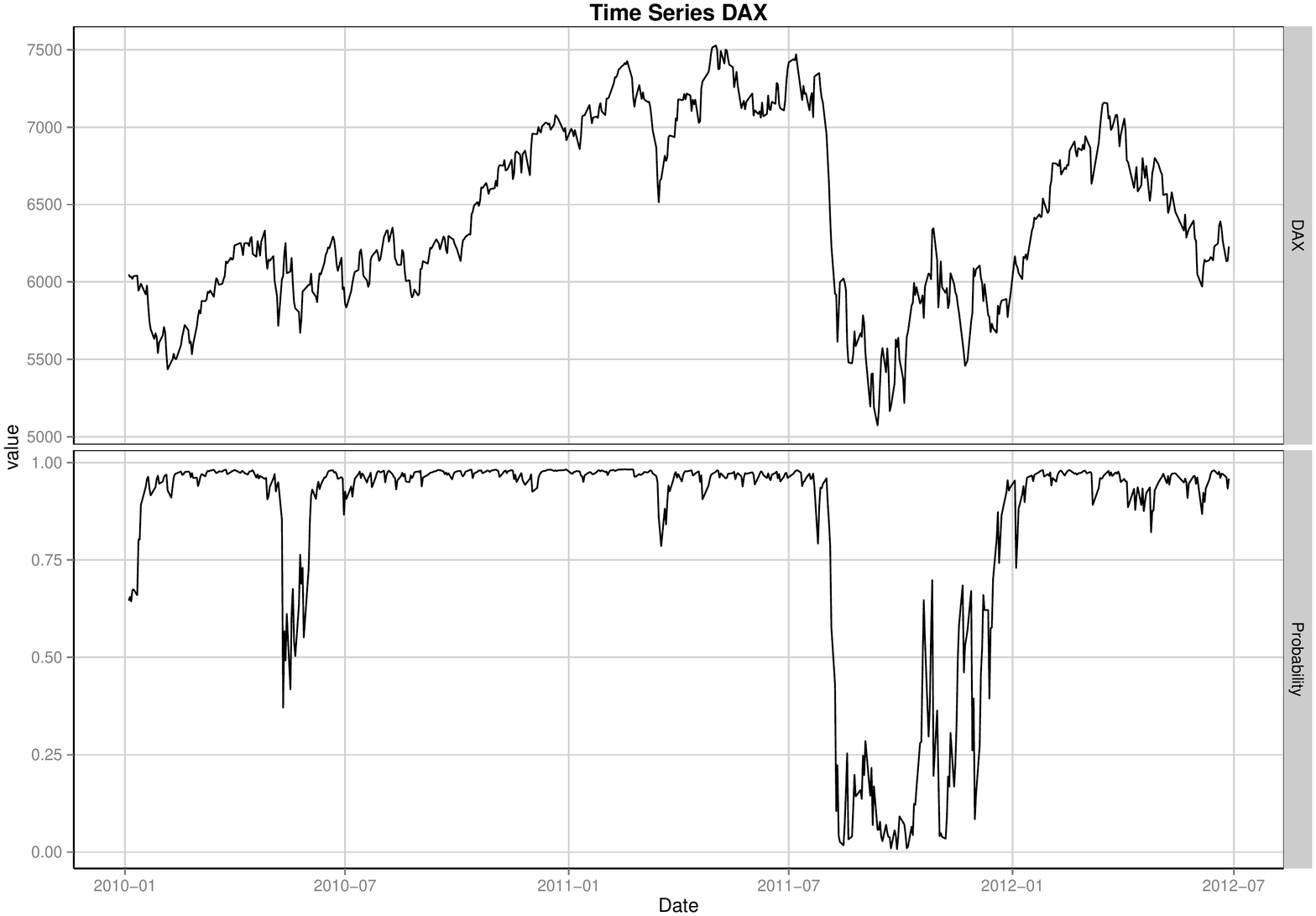}
\caption{Posterior probability - DAX index}
  \label{postprob}
\end{center}
\end{figure}
 
We see that the posterior probability reacts relatively quickly to losses in the index. In the downturn beginning in August 2011 the model captures this tense market state and swings into the pessimistic state. The posterior probability takes values close to zero or close to one most of the time telling us quite clearly in which state we are.
Observing this object could be of great value in risk management. Recall that the calibration period ended in $2009$ and that thus the data shown is purely out-of-sample. Despite this fact the posterior probability shows a satisfying pattern even in $2011$ which is already more than one year after the calibration period.

Additionally, the posterior probability measure is an input needed for the calculation of risk measures as we will see in the following subsection.

\subsection{Value-at-Risk}\label{Subsec:VaR}

In order to derive the value-at-risk (VaR) by Fourier methods we cite the characteristic function from~\cite{rogers}.

\begin{Definition}
The characteristic function of the N-period return distribution can be calculated by
\begin{align} \label{eq:chf}
\phi_N(t)=p_n\left(P\begin{pmatrix}\phi^1(t) &0 \\0 &\phi^2(t) \end{pmatrix}\right)^N\mathbf{1}
\end{align}
where $p_n$ is the posterior distribution at time $n$, given in~\eqref{eq:post}, $P$ the transition matrix and $\phi^i$, $i \in \{1,2\}$, are the univariate one-period characteristic functions of the asset return for the two states.
\end{Definition}

Note that by this definition the market returns of all assets that are modelled together influence the posterior view in $p_n$ and thereby the distribution of the specific asset considered. For example in the model that we calibrated, an increase in implied volatility as measured by VDAX will change this view and therefore change the posterior marginal distribution of the DAX.

For the purpose of backtesting the VaR below, we will stick to the case $N=1$. But the general form~\eqref{eq:chf} is needed for holding periods of more than one day. 

\begin{Definition}\label{defi:var}
The value-at-risk is defined by
\begin{align*}
VaR_\alpha=\inf\left\{x \in \rr: F(x)\geq \alpha\right\},
\end{align*}

where $\alpha$ is the confidence level and $F$ the cumulative distribution function of the asset return. If $F$ is continuous and strictly increasing the VaR is the unique solution of $F(VaR_\alpha)=\alpha$.
\end{Definition}

The general approach to the calculation of risk measures that we apply to this model uses Fourier techniques from option pricing theory
(see e.g. Lewis' method~\cite[Section~11.1.3]{cont2004financial} or~\cite{Lewis:simple}). To this aim we  extend the definition of the characteristic function $\phi$ to a function of a complex argument. Due to~\cite[Theorem 7.1.1]{Lukacs1970:characteristic} this complex Fourier transform is well defined on a strip of regularity in the complex plane, denoted by $S_X$.  
Let $f$ be the density function and $\phi$ the corresponding characteristic function.
Then by an application of the Plancherel (sometimes called Plancherel-Parseval) theorem for a function $g(x)$ it holds that
\begin{align}\label{eq:FourierApproach}
\int_{-\infty}^\infty f(x) g(x) dx = \frac{1}{2 \pi} \int_{-\infty+i \delta}^{\infty+i \delta} \phi(u) \hat{g}(u) du,
\end{align}
where $\hat{g}$ is the generalized Fourier transform of $g$ given by
\[
	\hat{g}(u) = \int_{-\infty+i \delta}^{\infty+i \delta} e^{-i u z} g(z) dz,
\]
and $\delta = \Im(z)$ for $z \in S_X$ is chosen such that the integral on the right hand side converges.

The final aim is to represented the risk measure of interest such that it can be efficiently calculated by the Fast Fourier Transform (FFT). For details of an implementation of FFT see e.g.~\cite{fft}.

\begin{Proposition}
Let $\phi$ be the characteristic function of the distribution function $F$ and $\delta>0$, then $F$ can be computed by the following formula
\begin{align}\label{eq:vaRF}
F(y)=\frac{e^{\delta y}}{2 \pi}  \int_{-\infty}^{\infty} e^{- i v y} \frac{\phi(v + i \delta)}{\delta- iv} dv.
\end{align}
\end{Proposition}

\begin{proof}
We can apply the general approach and find
\begin{align*}
\hat{g}(u) &= \int_{-\infty}^\infty e^{-i u x} \mathds{1}_{\{x \le y\}} dx \\
	       &= \frac{e^{-i u x}}{-iu} \Big|_{x = -\infty}^y = \frac{i e^{-i u y}}{u},
\end{align*}
for $\Im(u) = \delta >0$ such that the above integral converges.
Applying~\eqref{eq:FourierApproach} and a variable change we get
\begin{align*}
	F(y) &=  \frac{i}{2 \pi} \int_{-\infty+i \delta}^{\infty+i \delta}  e^{-i u y} \frac{\phi(u)}{u} du \overset{|v = u - i \delta|}{=} \\
	     &= \frac{e^{\delta y}}{2 \pi}  \int_{-\infty}^{\infty} e^{- i v y} \frac{\phi(v + i \delta)}{\delta- iv} dv.
\end{align*}

\end{proof}

The structure of~\eqref{eq:vaRF} allows us to compute the distribution function efficiently by FFT. Finally a simple root search delivers the value-at-risk. 

We apply this approach to the DAX index which is calibrated on data from November $2005$ to December $2009$ as described above and analyze the in-sample fit first. Figure~\ref{var} shows the index returns together with the $VaR_{95\%}$ estimate resulting from~\eqref{eq:vaRF}. The dashed horizontal lines are drawn at the $VaR_{95\%}$ in the two states respectively. Changes in the VaR estimate in the area between the dashed lines in Figure~\ref{var} are due to changes in the posterior distribution~\eqref{eq:post}, which is plotted below, reflecting the combined development of the market indices that are used in its calculation.

We backtest the VaR model using the binomial test~\cite{kup}. Observing $29$ violations with a sample of $999$ days gives a p-value of $0.999$ and thus strongly supports the accuracy of the model. 

To test the null hypothesis of random occurrences of VaR-violations we use the runs test as implemented in the R package tseries~\cite{rpackage:tseries} and find a p-value of $0.44$ which allows us to accept this null hypothesis. This confirms that the posterior distribution is reasonable sensitive to changes in the riskiness of the asset return.

\begin{figure}[H]
\begin{center}
\includegraphics[width=\textwidth]{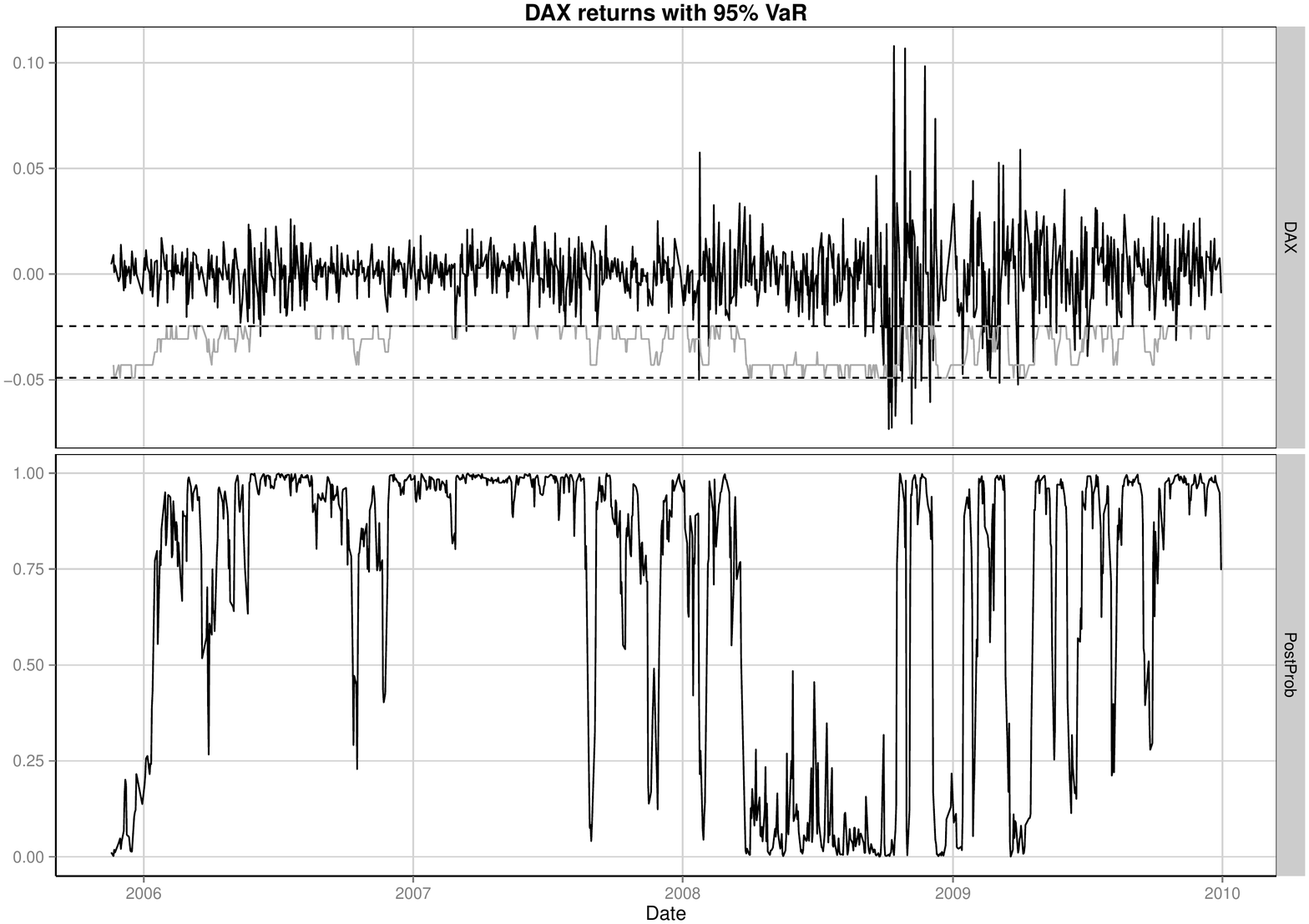}
\caption{$95\%$ value-at-risk for the DAX in-sample}
  \label{var}
\end{center}
\end{figure}

As a second risk measure we treat the calculation of expected shortfall in this model. Note that the distribution of asset returns is continuous. Thus we can state the following definition:

\begin{Definition}
The expected shortfall (ES) is defined by
\begin{align*}
ES_\alpha(X) = \ew{X|X \le VaR_\alpha(X)} ,
\end{align*}
where $\alpha$ is the confidence level and $VaR_\alpha$ is the value-at-risk to the level $\alpha$ as defined in Definition~\ref{defi:var}.
\end{Definition}

\begin{Proposition}
Let $\phi$ be the characteristic function of the distribution function $F$ and $\delta>0$, then the distribution function can be computed by the following formula
\begin{align}\label{eq:ESF}
ES_\alpha(X) = \frac{1}{1-\alpha} \frac{e^{\delta \operatorname{VaR}_{\alpha}}}{2 \pi}  \int_{-\infty}^{\infty} e^{- i v \operatorname{VaR}_{\alpha}} \frac{(\delta-i v)\operatorname{VaR}_{\alpha}-1}{(\delta-iv)^2} \phi(v + i \delta) dv.
\end{align}
\end{Proposition}

\begin{proof}
First we calculate 
\[
	\int_{-\infty}^y x f(x) dx,
\]
by the general approach~\eqref{eq:FourierApproach} above and then plug in the value-at-risk for $y$ and apply the FFT.
Applying integration by parts we get the transform
\begin{align*}
\hat{g}(u) &= \int_{-\infty}^\infty x e^{-i u x} \mathds{1}_{\{x \le y\}} dx \\
	       &= x \frac{e^{-i u x}}{-iu} \Big|_{x = -\infty}^y - \int_{-\infty}^y \frac{e^{-i u x}}{-iu} dx\\
	       &= y \frac{e^{-i u y}}{-i u} - \frac{e^{-i u y}}{(-i u)^2} = \frac{e^{-i u y}}{(-i u)^2} \left(-i u y -1\right),
\end{align*}
for $\Im(u) = \delta >0$ such that the above integral converges.
Applying~\eqref{eq:FourierApproach} and a variable change we get
\begin{align*}
	F(y) &=  \frac{1}{2 \pi} \int_{-\infty+i \delta}^{\infty+i \delta}  \frac{e^{-i u y}}{(-i u)^2} \left(-i u y -1\right) \phi(u) du \overset{|v = u - i \delta|}{=} \\
	     &= \frac{e^{\delta y}}{2 \pi}  \int_{-\infty}^{\infty} e^{- i v y} \frac{(\delta-i v)y-1}{(\delta-iv)^2} \phi(v + i \delta) dv.
\end{align*}
The result follows by plugging in the VaR for $y$ and using the definition of expected shortfall.
\end{proof}
The form of~\eqref{eq:ESF} allows for an application of the FFT. Reading off the value at the value-at-risk from the transformed sequence gives the expected shortfall.

Figure~\ref{ESoos} shows the index returns together with the $ES_{95\%}$ estimate resulting from~\eqref{eq:ESF} in the out-of-sample period. The dashed horizontal lines are drawn at the $ES_{95\%}$ in the two states respectively. Again we see the estimations of $ES_{95\%}$ wander up and down between these two extremes as the posterior probability changes. The $ES_{95\%}$ is breached only once in the out-of-sample period.

Concerning $VaR_{95\%}$ in the out-of-sample period we observe $10$ breaches during $610$ days which gives a p-value of $0.99$ in the binomial test. A runs test for the mixing of the breaches gives a p-value $0.6772$. These test results allow us to accept the model although the small number of breaches could indicate that the model calibrated in the highly volatile years around $2008$ is slightly too conservative for the out-of-sample period. In an industrial application the model should be recalibrated on a regular basis, e.g. once a month. Note that due to the high number of parameters the model is calibrated on $3$ years of data. Thus more frequent calibration will not change the model significantly as the influence of the data added is rather small.
\newline

In Figure~\ref{fig:VaRComaprison} we compare the value-at-risk derived from the regime switching model to the value-at-risk from a much simpler approach in order to see how much we gain by introducing the two states. As a comparative model we simply assume a GHYP distribution of the returns. Doing so we loose some of the features mentioned above but we still can model heavy tails and skewness. While we have used a history of approximately $3$ years for calibrating the regime switching model we need much less data for the \emph{simple} model and use this advantage. Namely, in each time step we look back for $250$ business days and fit the distribution and calculate the VaR. This VaR is then compared to the return of the following day. Figure~\ref{fig:VaRComaprison} shows that this simple model performs much worse. It reacts too slowly to changing market conditions and accumulates $36$ breaches in the testing period whereas the VaR of the regime switching model is breached only $10$ times. 
This illustrates the benefit of the modelling approach.

\begin{figure}[H]
\begin{center}
\includegraphics[width=\textwidth]{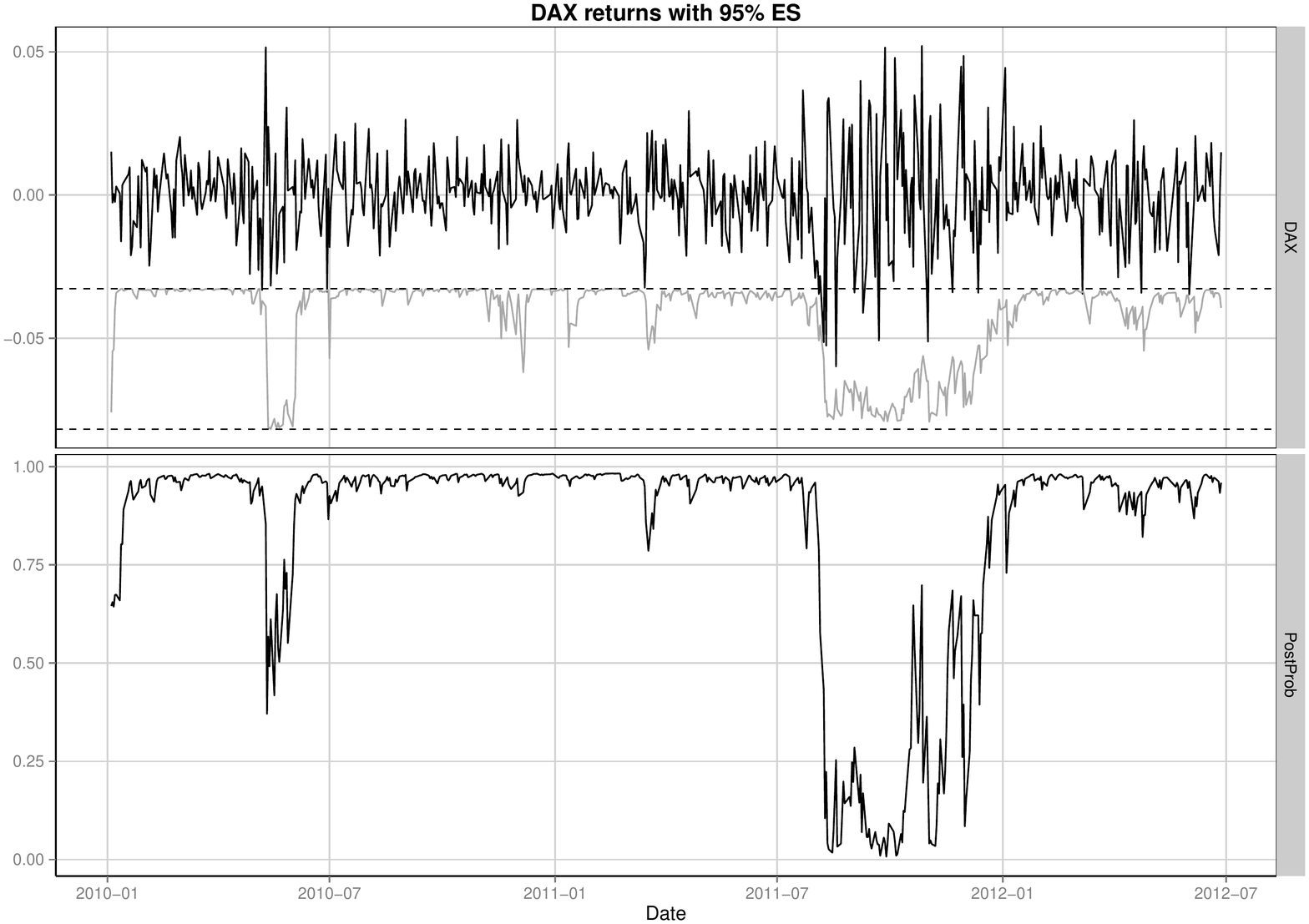}
\caption{$95\%$ expected shortfall with posterior probability}
  \label{ESoos}
\end{center}
\end{figure}

\begin{figure}[H]
\begin{center}
\includegraphics[width=\textwidth]{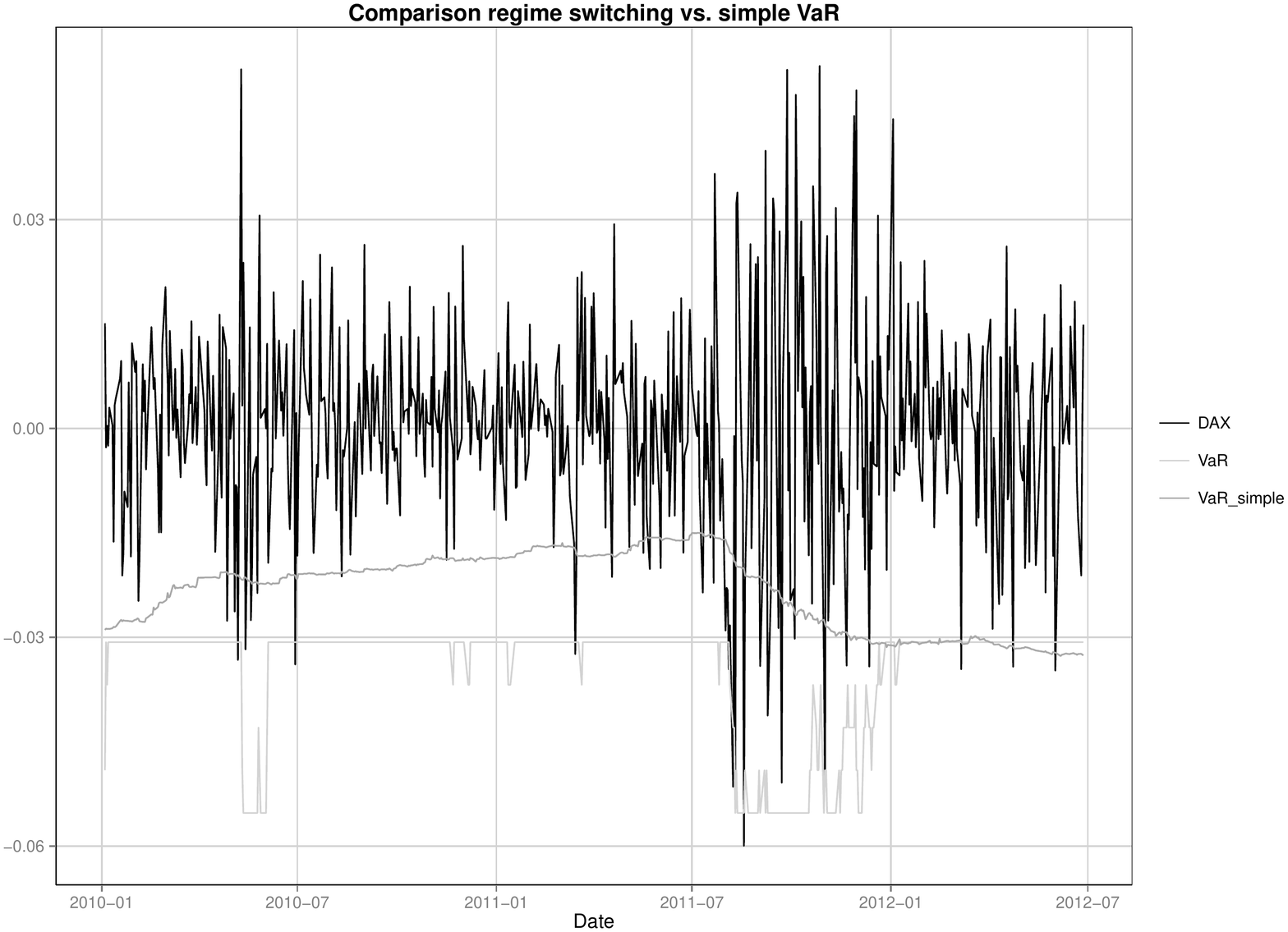}
\caption{$95\%$ VaR of the regime switching model compared to the \emph{simple} model}
  \label{fig:VaRComaprison}
\end{center}
\end{figure}

\section{Conclusion}

In this article the model introduced by Rogers and Zhang~\cite{rogers} is applied to calculation of the widely used risk measures value-at-risk and expected shortfall. In order to find a proper description of market conditions we calibrate the model to stocks, bonds, commodities and volatility. The volatility index is included in the calibration in order to increase the predictive power of the resulting model. Fourier methods from option pricing are applied to derive formulas that allow for an application of FFT to calculate the risk measures. 
Finally in-sample and out-of-sample analysis show the benefit of the method proposed.

Although the calibration boils down to a maximum-likelihood estimation a proper choice of the length of the calibration period turns out to be crucial which can be seen as a weak point of the method. 

\bibliographystyle{plain}
\bibliography{References}

\end{document}